\documentclass[12pt]{article}
\usepackage[latin1]{inputenc}
\usepackage{lmodern}

\usepackage{amssymb, amsmath, amsthm}
\usepackage[a4paper,top=25mm,bottom=25mm,left=25mm,right=25mm]{geometry}
\usepackage{etex}

\usepackage{authblk} 
\usepackage{pifont}
\usepackage{graphicx}
\usepackage[usenames,dvipsnames,svgnames,table]{xcolor}
\usepackage[figuresright]{rotating}
\usepackage{xtab} 
\usepackage{longtable} 
\usepackage{footnote}
\usepackage[stable]{footmisc}
\usepackage{chngpage} 
\usepackage{pdflscape} 

\usepackage{pgfplots}
\usepackage{setspace}

\makesavenoteenv{tabular}
\usepackage{tabularx}
\usepackage{booktabs}
\usepackage{tabu}
\usepackage{threeparttable}
\usepackage[referable]{threeparttablex} 
\newcolumntype{R}{>{\raggedleft\arraybackslash}X}
\newcolumntype{L}{>{\raggedright\arraybackslash}X}
\newcolumntype{C}{>{\centering\arraybackslash}X}
\newcolumntype{A}{>{\columncolor{gray!25}}C}
\newcolumntype{a}{>{\columncolor{gray!25}}c}

\usepackage{dcolumn} 
\newcolumntype{.}{D{.}{.}{-1}}

\usepackage{tikz}
\usetikzlibrary{arrows}
\usepackage[semicolon]{natbib}
\PassOptionsToPackage{hyphens}{url}\usepackage[breaklinks]{hyperref} 
\hypersetup{
  colorlinks   = true,    
  urlcolor     = blue,    
  linkcolor    = blue,    
  citecolor    = red      
}
\usepackage{microtype}
\usepackage[justification=centering]{caption}

\usepackage[labelformat=simple]{subcaption}

\DeclareCaptionLabelFormat{parenthesis}{(#2)}
\makeatletter
\renewcommand\p@subfigure{\arabic{figure}.}
\makeatother

\DeclareCaptionLabelFormat{parenthesis}{(#2)}
\makeatletter
\renewcommand\p@subtable{A.\arabic{table}.}
\makeatother

\usepackage{enumitem}

\setlist[itemize]{leftmargin=3\parindent}
\setlist[enumerate]{leftmargin=2\parindent}

\theoremstyle{plain}

\newtheorem{lemma}{Lemma}
\newtheorem{proposition}{Proposition}
\newtheorem{theorem}{Theorem}

\theoremstyle{definition}

\newtheorem{definition}{Definition}
\newtheorem{example}{Example}

\theoremstyle{remark}
\newtheorem{notation}{Notation}
\newtheorem{remark}{Remark}

\def\keywords{\vspace{.5em} 
{\textit{Keywords}:\,\relax%
}}

\def\JEL{\vspace{.5em} 
{\textbf{JEL classification number}:\,\relax%
}}

\def\AMS{\vspace{.5em} 
{\textbf{AMS classification number}:\,}}

\author{L\'aszl\'o Csat\'o\thanks{~e-mail: laszlo.csato@uni-corvinus.hu} }
\affil{Department of Operations Research and Actuarial Sciences \\ Corvinus University of Budapest \\ MTA-BCE ''Lend\"ulet'' Strategic Interactions Research Group \\ Budapest, Hungary}
\title{Distance-based accessibility indices\thanks{~We are grateful to Dezs\H{o} Bednay for useful advices.
\newline
The research was supported by OTKA grant K 111797 and MTA-SYLFF (The Ryoichi Sasakawa Young
Leaders Fellowship Fund) grant 'Mathematical analysis of centrality measures', awarded to the author in
2015.}}
\date{\today}

\begin{document}

\maketitle

\begin{abstract}
The paper attempts to develop a suitable accessibility index for networks where each link has a value such that a smaller number is preferred like distance, cost, or travel time.
A measure called distance sum is characterized by three independent properties: anonymity, an appropriately chosen independence axiom, and dominance preservation, which requires that a node not far to any other is at least as accessible.
 
We argue for the need of eliminating the independence property in certain applications. Therefore generalized distance sum, a family of accessibility indices, will be suggested. It is linear, considers the accessibility of vertices besides their distances and depends on a parameter in order to control its deviation from distance sum.
Generalized distance sum is anonymous and satisfies dominance preservation if its parameter meets a sufficient condition. Two detailed examples demonstrate its ability to reflect the vulnerability of accessibility to link disruptions.

\JEL{D85, Z13}

\AMS{15A06, 91D30}

\end{abstract}

\section{Introduction} \label{Sec1}

Section~\ref{Sec1} aims to confirm the importance of accessibility indices. Some applications of them are presented in Subsection~\ref{Sec11} together with a detailed literature review in Subsection~\ref{Sec12}. Finally, Subsection~\ref{Sec13} gives an outline of our approach and main results.

\subsection{Motivation} \label{Sec11}

A key and basic concept in network analysis that researchers want to capture is centrality. The well-known classification of \citet{Freeman1979} distinguishes three conceptions, measured by degree, closeness, and betweenness, respectively. However, there is a frequent need for other centrality models.
One of them, mainly used in theoretical geography for analysing social activities and regional economies, can be called \emph{accessibility}. Accessibility index provides a numerical answer to questions such as 'How accessible is a node from other nodes in a network?' or 'What is its relative geographical importance?'.

Accessibility measures have a number of interesting ways of utilization:
\begin{enumerate}
\item
Knowledge of which nodes have the highest accessibility could be of interest in itself (e.g. by revealing their strategic importance);
\item
The accessibility of vertices could be statistically correlated to other economic, sociological or political variables;
\item
Accessibility of the same nodes (e.g. urban centres) in different (e.g. transportation, infrastructure) networks could be compared;
\item
Proposed changes in a network could be evaluated in terms of their effect on the accessibility of vertices;
\item
Networks (e.g. empires) could be compared by their propensity to disintegrate. For example, it may be difficult to manage from a unique centre if the most accessible nodes are far from each other. 
\end{enumerate}

Practical application of accessibility involve, among others, an analysis of the medieval river trade network of Russia \citep{Pitts1965, Pitts1979}, of the medieval Serbian oecumene \citep{Carter1969}, of the interstate highway network for cities in the southeastern United States \citep{Garrison1960, MackiewiczRatajczak1996}, of the inter-island voyaging network of the Marshall islands \citep{HageHarary1995}, or of the global maritime container transportation network \citep{WangCullinane2008}.
Accessibility indices can be used in operations research, too, where a typical problem is that of choosing a site for a facility on the basis of a specific criterion \citep{Slater1981}.

\subsection{Some methods for measuring accessibility} \label{Sec12}

The adjacency matrix $C$ of the network graph is given by $c_{ij} = 1$ if and only if nodes $i$ and $j$ are connected and $c_{ij} = 0$ otherwise.

One of the first accessibility indices, sometimes called \emph{connection array}, was introduced by \citet{Garrison1960}. It is based on powering of matrix $C$ such that
\[
T = C + C^2 + C^3 + \dots + C^m = \sum_{i=1}^m C^i.
\]
The accessibility of a node comes from summing the corresponding column of matrix $T$, i.e. the total number of at most $m$-long paths to other nodes. It was used by \citet{Pitts1965} and \citet{Carter1969} despite two significant shortcomings:
\begin{itemize}
\item
The increase of $m$ distorts reality as some redundant paths (i.e. longer than the shortest path in a topological sense) are included in the calculation of accessibility.
\item
It is not obvious what is the appropriate value of $m$. The usual choice is the \emph{diameter}, the
maximum number of edges in the shortest path between the furthest pair of nodes, which provides that $T$ has only positive elements (if $m \geq 2$).
\end{itemize}
The first problem may be addressed by introducing a multiplier for $C^i$ showing exponential decay in the calculation of matrix $T$, as suggested by \citet{Garrison1960} following \citet{Katz1953}. However, it is still not able to completely overcome the inherent failures of the method \citep{MackiewiczRatajczak1996}.

In order to improve the comparability of results achieved for various $m$, \citet{Stutz1973} proposed a formula to determine the relative accessibility of nodes.

Naturally, one can focus only on the shortest paths between the nodes. Then a plausible measure is the total length of them to all other nodes \citep{Shimbel1951, Shimbel1953, Harary1959, Pitts1965, Carter1969}, which is called \emph{shortest-path array}.

\emph{Betweenness} measures how many shortest paths pass through the given node \citep{Shimbel1953, Freeman1977}. It was also applied as an accessibility index \citep{Pitts1979}.

\emph{Eccentricity} concentrates on the maximum distance to all other nodes. Sometimes it is the suitable concept for predicting the politically and symbolically important nodes in a network \citep{HageHarary1995}.

A whole branch of literature, initiated by \citep{Gould1967}, propagates the use of eigenvectors for measuring accessibility \citep{Carter1969, Tinkler1972, Straffin1980, MackiewiczRatajczak1996, WangCullinane2008}.
It is based on the Perron-Frobenius theorem \citep{Perron1907, Frobenius1908, Frobenius1909, Frobenius1912} for nonnegative square matrices: matrix $C$ has a unique (up to constant multiples) positive eigenvector corresponding to the principal eigenvalue. It have been proposed that the non-principal eigenvectors might also have geographical meaning, however, its use remains controversial \citep{Carter1969, Hay1975, Tinkler1975, Straffin1980}.

Finally, \citet{AmerGimenezMagana2012} construct accessibility indices to the nodes of directed graphs using methods of game theory.

\subsection{Aims and tools} \label{Sec13}

The focus on the adjacency matrix has no sense in some applications, where the network graph is closely complete \citep{WangCullinane2008}.
However, as \citet{TaaffeGauthier1973} claim, it is not necessary to restrict the concept to a purely topological one. Value of linkages may be given as a distance between the nodes in actual mileage as well as the cost of movement or the time required to travel between them. Actually, it was implemented by \citet{Carter1969} and \citet{Pitts1979} using the distance as a value.

We will call the edge values \emph{distance} throughout the paper. It may be any measure satisfying triangle inequality such that a smaller value is better. It can be assumed without loss of generality that the network is complete, i.e. a value is available for each link.
While the extension of shortest-path array to this domain, said to be the \emph{distance sum}, is obvious, the applicability of other measures -- presented in Subsection~\ref{Sec12} -- is ambiguous, despite the use of the principal eigenvector by \citet{WangCullinane2008}.

The paper aims to analyse accessibility indices in these networks.
For this purpose the axiomatic approach, a standard path in game and social choice theory will be applied. Distance sum will be characterized, that is, a set of properties will be given such that it is the only accessibility index satisfying them, but eliminating any axiom allows for other measures. According to our knowledge, no axiomatizations of accessibility indices exist. However, some results are provided for other centrality measures \citep{MonsuurStorcken2004, Garg2009, Kitti2012, DequiedtZenou2014}.

The characterization of distance sum, one of the main results, contains three axioms: \emph{anonymity}, an appropriately chosen independence axiom (\emph{independence of distance distribution}), and \emph{dominance preservation}, which requires that a node not far to any other is at least as accessible. Furthermore, it will be presented through an example that independence of distance distribution is a property one would rather not have in certain cases.
Therefore a parametric family of accessibility indices will be proposed.

It will be called \emph{generalized distance sum} as it redistributes the pool of aggregated distances (i.e. the sum of distances over the network) by considering distance distribution: in the case of two nodes with the same distance sum, it is deemed better to be close to more accessible nodes (and therefore, to be far to less accessible ones) than vice versa.
It also takes into account -- similarly to connection array -- longer paths than the shortest, their impact is governed by the parameter. The main argument for their inclusion is that shortest paths may be vulnerable to disruptions, when the value of alternative routes becomes relevant. The limit of generalized distance sum is the distance sum.

Generalized distance sum satisfies anonymity and violates independence of distance distribution.
The other main result provides a sufficient condition for this measure -- by limiting the value of its parameter -- to meet dominance preservation, the third axiom in the characterization of distance sum.
However, it will not be axiomatized. 

While it is not debated that characterizations are a correct way to distinguish between accessibility indices, we think they have limited significance for applications: if one should determine the accessibilities in a \emph{given} network, he/she is not much interested in the properties of the measure on other networks. Characterizations could reveal some aspects of the choice but the consequences of the axioms on the actual network remain obscure.

The axiomatic point of view is not exclusive. \citet{BorgattiEverett2006} criticize this way, taken by \citet{Sabidussi1966}, because it does not 'actually attempt to explain what centrality is'.
Thus generalized distance sum will be scrutinized through some examples, showing that it is able to reflect the structure of the network. 

The paper is structured as follows. Section~\ref{Sec2} describes the characterization of distance sum. Subsection~\ref{Sec21} defines the model and the accessibility index. Subsection~\ref{Sec22} is devoted to introduce the axioms and to prove the theorem as well as to discuss the independence of properties and an extension of the result.
Section~\ref{Sec3} presents a new parametric measure of accessibility by challenging an axiom of the previous characterization. Generalized distance sum is motivated and introduced in Subsection~\ref{Sec31}. We also deal with some of its properties and return to address triangle inequality, a potential constraint of the model. Subsection~\ref{Sec32} gives a sufficient condition for generalized distance sum to satisfy another property used in the axiomatization. It is shown that an excessive value of the parameter may lead to counter-intuitive accessibility rankings. Two examples are discussed in Subsection~\ref{Sec33} to reveal some interesting feature of the suggested measure.
Sections~\ref{Sec2} and \ref{Sec3} contain somewhat independent results, the axiomatization of distance sum is wholly autonomous, and generalized distance sum may be submitted and discussed without the characterization, too.
Finally, Section~\ref{Sec4} summarizes our findings and outlines some directions for future research.

\section{A characterization of distance sum} \label{Sec2}

In this section, some natural conditions will be conceived for accessibility indices, similarly to centrality measures \citep{Sabidussi1966, Chienetal2004, LandherrFriedlHeidemann2010, BoldiVigna2014}.
Distance sum, an obvious solution for our problem, will be characterized, that is, a set of properties will be given such that it is the only accessibility index satisfying them, while eliminating any axiom allows for other measures.

\subsection{The model} \label{Sec21}

\begin{definition} \label{Def1}
\emph{Transportation network}:
\emph{Transportation network} is a pair $(N,D)$ such that
\begin{itemize}
\item
$N = \{ 1,2, \dots ,n \}$ is a finite set of nodes;
\item
$D \in \mathbb{R}_+^{n \times n}$ is a symmetric distance matrix satisfying triangle inequality, namely, $d_{ij} \leq \sum_{\ell = 0}^{m-1} d_{k_\ell k_{\ell + 1}}$ where $(i = k_0, k_1, \dots ,k_m = j)$ is a path between the nodes $i$ and $j$.
$d_{ii} = 0$ for all $i = 1,2,\dots ,n$.
\end{itemize}
\end{definition}

$D$ can be the adjacency matrix of a complete, weighted, undirected graph (without loops and multiple edges) with the set of nodes $N$.\footnote{~Assumption of completeness is not restrictive in the case of connected graphs since the distance of any pair of nodes can be measured by the shortest path between them.}

\begin{notation} \label{Not1}
$\mathcal{N}^n$ is the class of all transportation networks $(N,D)$ with $|N| = n$.
\end{notation}

\begin{definition} \label{Def2}
\emph{Accessibility index}:
Let $(N,D) \in \mathcal{N}^n$ be a transportation network. \emph{Accessibility index} $f$ is a function that assigns an $n$-dimensional vector of real numbers to $(N,D)$ with $f_i(N,D)$ showing the accessibility of node $i$.
\end{definition}

\begin{notation} \label{Not2}
$f: \mathcal{N}^n \to \mathbb{R}^n$ is an accessibility index.
\end{notation}

Node $i$ is said to be at least as accessible as node $j$ in the transportation network $(N,D)$ if and only if $f_i(N,D) \leq f_j(N,D)$, so a smaller value of accessibility is more favourable.

\begin{definition} \label{Def3}
\emph{Order equivalence}:
Let $f,g: \mathcal{N}^n \to \mathbb{R}^n$ be two accessibility indices. They are called \emph{order equivalent} if and only if $f_i(N,D) \leq f_j(N,D) \Leftrightarrow f_i(N,D) \leq f_j(N,D)$ for any transportation network $(N,D) \in \mathcal{N}^n$. .
\end{definition}

\begin{notation} \label{Not3}
$f \approx g$ means that accessibility measures $f,g: \mathcal{N}^n \to \mathbb{R}^n$ are order equivalent.
\end{notation}

Since we focus on accessibility rankings, order equivalent accessibility indices are not distinguished. For example, an accessibility index is invariant under multiplication by positive scalars.

Throughout the paper, vectors are denoted by bold fonts and assumed to be column vectors.
Let $\mathbf{e} \in \mathbb{R}^n$ be the column vector such that $e_i = 1$ for all $i = 1,2, \dots ,n$ and $I \in \mathbb{R}^{n \times n}$ be the matrix with $I_{ij} = 1$ for all $i,j = 1,2, \dots ,n$.

The first, almost trivial idea to measure accessibility can be the sum of distances to all other nodes. In fact, it is extensively used in the literature \citep{Shimbel1951, Shimbel1953, Harary1959, Pitts1965, Carter1969}.

\begin{definition} \label{Def4}
\emph{Distance sum}: $\mathbf{d^\Sigma}: \mathcal{N}^n \to \mathbb{R}^n$ such that $d_i^\Sigma = \sum_{j \in N} d_{ij}$ for all $i \in N$.
\end{definition}

\subsection{Axioms and characterization} \label{Sec22}

The first condition is independence of the labelling of the nodes.

\begin{definition} \label{Def5}
\emph{Anonymity} ($ANO$):
Let $(N,D),(\sigma N, \sigma D) \in \mathcal{N}$ be two transportation networks such that $(\sigma N, \sigma D)$ is given by a permutation of nodes $\sigma: N \rightarrow N$ from $(N,D)$.
Accessibility index $f: \mathcal{N}^n \to \mathbb{R}^n$ is called \emph{anonymous} if $f_i(N,D) = f_{\sigma i}(\sigma N, \sigma D)$ for all $i \in N$.
\end{definition}

Property $ANO$ implies that two symmetric nodes are equally accessible. It also ensures that all nodes have the same accessibility in a transportation network with the same distance between any pair of nodes.

\begin{lemma} \label{Lemma1}
Distance sum satisfies $ANO$.
\end{lemma}

While distance sum is usually a good baseline to approximate accessibility, it does not consider the distribution of distances.

\begin{definition} \label{Def6}
\emph{Independence of distance distribution} ($IDD$):
Let $(N,D) \in \mathcal{N}^n$ be a transportation network and $i,j,k,\ell \in N$ be four distinct nodes.
Let $f: \mathcal{N}^n \to \mathbb{R}^n$ be an accessibility index such that $f_i(N,D) \leq f_j(N,D)$ and $(N,D') \in \mathcal{N}^n$ be a transportation network identical to $(N,D)$ except for $d'_{ik} \neq d_{ik}$ and $d'_{i \ell} \neq d_{i \ell}$ but $d'_{ik} + d'_{i \ell} = d_{ik} + d_{i \ell}$. \\
$f$ is called \emph{independent of distance distribution} if $f_i(N,D') \leq f_j(N,D')$.
\end{definition}

Property $IDD$ implies that the accessibility ranking does not change if the distance sum of each node remains the same.

\begin{lemma} \label{Lemma2}
Distance sum satisfies $IDD$.
\end{lemma}

\begin{definition} \label{Def7}
\emph{Dominance}:
Let $(N,D) \in \mathcal{N}^n$ be a transportation network and $i,j \in N$ be two distinct nodes such that $d_{ik} \leq d_{jk}$ for all $k \in N \setminus \{ i,j \}$ with a strict inequality ($<$) for at least one $k$. \\
Then it is said that node $i$ \emph{dominates} node $j$.
\end{definition}

A natural requirement for accessibility indices can be that the accessibility of 'obviously' more external nodes is worse.

\begin{definition} \label{Def8}
\emph{Dominance preservation} ($DP$):
Let $(N,D) \in \mathcal{N}^n$ be a transportation network and $i,j \in N$ be two distinct nodes such that node $i$ dominates node $j$. \\
Accessibility index $f: \mathcal{N}^n \to \mathbb{R}^n$ \emph{preserves dominance} if $f_i(N,D) < f_j(N,D)$.
\end{definition}

Property $DP$ demands that if node $j$ is at least as far from all nodes as node $i$, then it has a larger accessibility value provided that they could not be labelled arbitrarily.
A similar axiom called \emph{self-consistency} has been suggested by \citep{ChebotarevShamis1997b_rus} for preference aggregating methods.

\begin{lemma} \label{Lemma3}
Distance sum satisfies $DP$.
\end{lemma}

Requirements $ANO$, $IDD$ and $DP$ still give a characterization of distance sum.

\begin{theorem} \label{Theo1}
If an accessibility index $f: \mathcal{N}^n \to \mathbb{R}^n$ is anonymous, independent of distance distribution and preserves dominance, then $f$ is the distance sum.
\end{theorem}

\begin{proof}
Consider a transportation network $(N,D) \in \mathcal{N}^n$ and three nodes $i,j,k \in N$. It can be assumed that $\sum_{m \in N} d_{im} \leq \sum_{m \in N} d_{jm}$. Define the transportation network $(N,D') \in \mathcal{N}^n$ such that $d'_{im} = d_{im}$ for all $m \in N$, $d'_{jk} = d'_{ik} + \left( \sum_{m \in N} d_{jm} - \sum_{m \in N} d_{im} \right) \leq d'_{ik}$ and $d'_{j \ell} = d'_{i \ell}$ for all $\ell \in N \setminus \{ i,j,k \}$. Two cases can be distinguished:
\begin{itemize}
\item
$\sum_{m \in N} d_{im} < \sum_{m \in N} d_{jm}$

Then $d'_{ik} < d'_{jk}$, so $f_i(N,D') < f_j(N,D')$ due to $DP$ and $f_i(N,D) < f_j(N,D)$ because of $IDD$ as $\sum_{m \in N} d'_{im} = \sum_{m \in N} d_{im}$ and $\sum_{m \in N} d'_{jm} = \sum_{m \in N} d_{jm}$.

\item
$\sum_{m \in N} d_{im} = \sum_{m \in N} d_{jm}$

Then $d'_{ik} = d'_{jk}$, so $f_i(N,D') = f_j(N,D')$ due to $ANO$ and $f_i(N,D) = f_j(N,D)$ because of $IDD$ as $\sum_{m \in N} d'_{im} = \sum_{m \in N} d_{im}$ and $\sum_{m \in N} d'_{jm} = \sum_{m \in N} d_{jm}$.
\end{itemize}

It implies that $f \approx \mathbf{d^\Sigma}$. 
\end{proof}

\begin{remark} \label{Rem1}
The proof of Theorem~\ref{Theo1} requires at least three nodes. However, $IDD$ has a meaning only for four nodes. It can be seen that $ANO$ and $DP$ are enough to characterize distance sum if $n=3$ since $d_{ik} < d_{jk} \Rightarrow f_i(N,D) < f_j(N,D)$ because of dominance preservation and $d_{ik} = d_{jk} \Rightarrow f_i(N,D) = f_j(N,D)$ due to anonymity.
Note that $d_{ij} \leq d_{ik} \Leftrightarrow \sum_{m \in N} d_{im} \leq \sum_{m \in N} d_{jm}$ by definition.
\end{remark}

Now the independence of the axioms is addressed.

\begin{proposition} \label{Prop1}
Anonymity, independence of distance distribution and domination preservation are logically independent conditions.
\end{proposition}

\begin{proof}
As usual, an example will be given that satisfies all the axioms mentioned in Theorem~\ref{Theo1} except for the
one under consideration. Consider the following three functions:

\begin{definition} \label{Def9}
\emph{Distance sum without ties}: $\mathbf{d^{\Sigma <}}: \mathcal{N}^n \to \mathbb{R}^n$ such that $d_i^{\Sigma <} = \sum_{j \in N} d_{ij} + (i-1) \min \left\{ \sum_{\ell \in N} \left( d_{j \ell} - d_{k \ell} \right): j,k \in N \right\} / n$ for all $i \in N$.
\end{definition}

\begin{definition} \label{Def10}
\emph{Inverse distance sum}: $\mathbf{d^{-\Sigma}}: \mathcal{N}^n \to \mathbb{R}^n$ such that $d_i^{-\Sigma} = - \sum_{j \in N} d_{ij}$ for all $i \in N$.
\end{definition}

\begin{definition} \label{Def11}
\emph{Distance product}: $\mathbf{d^\Pi}: \mathcal{N}^n \to \mathbb{R}^n$ such that $d_i^\Pi = \prod_{j \neq i} d_{ij}$ for all $i \in N$.
\end{definition}

Distance sum without ties satisfies $IDD$ and $DP$ but violates $ANO$ as in the case of equal distance sums, the node with a lower index becomes more accessible.

Inverse distance sum satisfies $ANO$ and $IDD$ but violates $DP$.

Distance product satisfies $ANO$ and $DP$ but it is not equivalent to distance sum.

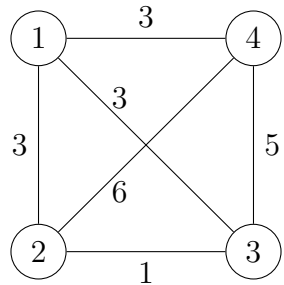
\begin{figure}[!ht]
\centering
\caption{Transportation network of Example \ref{Examp1}}
\label{Fig1}

\begin{tikzpicture}[scale=1,auto=center, transform shape, >=triangle 45]
\tikzstyle{every node}=[draw,shape=circle];
  \node (n1) at (135:2) {$1$};
  \node (n2) at (225:2) {$2$};
  \node (n3) at (315:2)  {$3$};
  \node (n4) at (45:2) {$4$};

  \path[every node/.style={}]
    (n1) edge node[pos=0.35,above] {3} (n3)
    (n2) edge node[pos=0.35,below] {6} (n4)
    (n1) edge node[above] {3} (n4)
    (n2) edge node[below] {1} (n3)
    (n1) edge node[left]  {3} (n2)
    (n3) edge node[right] {5} (n4);
\end{tikzpicture}
\end{figure}

\begin{example} \label{Examp1}
Consider the transportation network $(N,D) \in \mathcal{N}^4$ in Figure~\ref{Fig1} where
\[
\mathbf{d^\Sigma}(N,D) =
\left[
\begin{array}{cccc}
    9     & 10     & 9     & 14 \\
\end{array}
\right]^\top
\qquad
\text{and}
\qquad
\mathbf{d^\Pi}(N,D) =
\left[
\begin{array}{cccc}
    27     & 18     & 15     & 90  \\
\end{array}
\right]^\top.
\]
Note that $d_1^\Sigma(N,D) < d_2^\Sigma(N,D)$ and $d_1^\Pi(N,D) > d_2^\Pi(N,D)$.
\end{example}
\end{proof}

\begin{remark} \label{Rem2}
Theorem~\ref{Theo1} and Proposition \ref{Prop1} reveal a possible way to characterize inverse distance sum by $ANO$, $IDD$ and 'inverse' domination preservation (i.e. with an opposite implication compared to $DP$). \\
On the other hand, distance product can be axiomatized by $ANO$, a modification of $IDD$ such that the product of distances -- instead of their sum -- remains unchanged, and $DP$.
\end{remark}

\section{The suggested measure: generalized distance sum} \label{Sec3}

From the three axioms of Theorem~\ref{Theo1}, independence of distance distribution seems to be the less plausible. Therefore it is worth to examine its substitution by other considerations.
In this section a new accessibility index will be introduced and analysed with respect to its properties, especially dominance preservation.

\subsection{A new accessibility index} \label{Sec31}

$IDD$ could be eliminated by the use distance product, for instance.
However, the major disadvantage of distance sum may be that indirect connections are not taken into account, captured by the following condition. 

\begin{definition} \label{Def12}
\emph{Independence of irrelevant distances} ($IID$):
Let $(N,D) \in \mathcal{N}^n$ be a transportation network and $i,j,k,\ell \in N$ be four distinct nodes. Let $f: \mathcal{N}^n \to \mathbb{R}^n$ be an accessibility index such that $f_i(N,D) \leq f_j(N,D)$ and $(N,D') \in \mathcal{N}^n$ be a transportation network identical to $(N,D)$ except for $d'_{k \ell} \neq d_{k \ell}$. \\ 
$f$ is called \emph{independent of irrelevant distances} if $f_i(N,D') \leq f_j(N,D')$.
\end{definition}

It means that the relative accessibility of two nodes is not affected by the distances between other nodes.
Independence of irrelevant distances is an adaptation of the axiom \emph{independence of irrelevant matches} defined for (general) tournaments \citep{Rubinstein1980, Gonzalez-DiazHendrickxLohmann2013}.

\begin{lemma} \label{Lemma4}
Distance sum, distance sum without ties, inverse distance sum and distance product meet $IID$.
\end{lemma}

\begin{definition} \label{Def13}
\emph{Shortest path-based accessibility index}:
Let $(N,D) \in \mathcal{N}^n$ be any transportation network. \\
Accessibility index $f: \mathcal{N}^n \to \mathbb{R}^n$ is called a \emph{shortest path-based} accessibility index if it satisfies independence of irrelevant distances.
\end{definition}

\begin{remark} \label{Rem3}
Distance sum, distance sum without ties, inverse distance sum and distance product are shortest path-based accessibility indices.
\end{remark}

A shortest path-based accessibility index considers only the local structure of the network that is, distances to the other nodes. They can be interpreted as shortest paths due to triangle inequality.
Nevertheless, it is not always enough to focus on shortest paths between the nodes. It can occur that a disruption of some links forces the use of other paths. They may serve as optional detours, hence influencing accessibility.

\begin{figure}[!ht]
\centering
\caption{Transportation networks of Example~\ref{Examp2}}
\label{Fig2}
  
\begin{subfigure}{0.49\textwidth}
  \centering
  \subcaption{Transportation network $(N,D)$}
  \label{Fig2a}
\begin{tikzpicture}[scale=1,auto=center, transform shape, >=triangle 45]
\tikzstyle{every node}=[draw,shape=circle];
  \node (n1) at (144:3) {$1$};
  \node (n2) at (216:3) {$2$};
  \node (n3) at (72:3)  {$3$};
  \node (n4) at (288:3) {$4$};
  \node (n5) at (0:3)   {$5$};

  \path[every node/.style={}]
    (n1) edge node[above] {1} (n3)
    (n2) edge node[below] {1} (n4)
    (n1) edge node[left] {2} (n4)
    (n2) edge node[left] {2} (n3)
    (n3) edge node[right] {3} (n5)
    (n4) edge [blue,dashed] node[right] {3} (n5)
    (n1) edge node[above] {4} (n5)
    (n2) edge node[below] {4} (n5)
    (n1) edge node[left] {2} (n2)
    (n3) edge node[right] {2} (n4);
\end{tikzpicture}
\end{subfigure}
\begin{subfigure}{0.49\textwidth}
  \centering
  \subcaption{Transportation network $(N,D')$}
  \label{Fig2b}
\begin{tikzpicture}[scale=1,auto=center, transform shape, >=triangle 45]
\tikzstyle{every node}=[draw,shape=circle];
  \node (n1) at (144:3) {$1$};
  \node (n2) at (216:3) {$2$};
  \node (n3) at (72:3)  {$3$};
  \node (n4) at (288:3) {$4$};
  \node (n5) at (0:3)   {$5$};

  \path[every node/.style={}]
    (n1) edge node[above] {1} (n3)
    (n2) edge node[below] {1} (n4)
    (n1) edge node[left] {2} (n4)
    (n2) edge node[left] {2} (n3)
    (n3) edge node[right] {3} (n5)
    (n4) edge [blue,dashed] node[right] {4} (n5)
    (n1) edge node[above] {4} (n5)
    (n2) edge node[below] {4} (n5)
    (n1) edge node[left] {2} (n2)
    (n3) edge node[right] {2} (n4);
\end{tikzpicture}
\end{subfigure}
\end{figure}
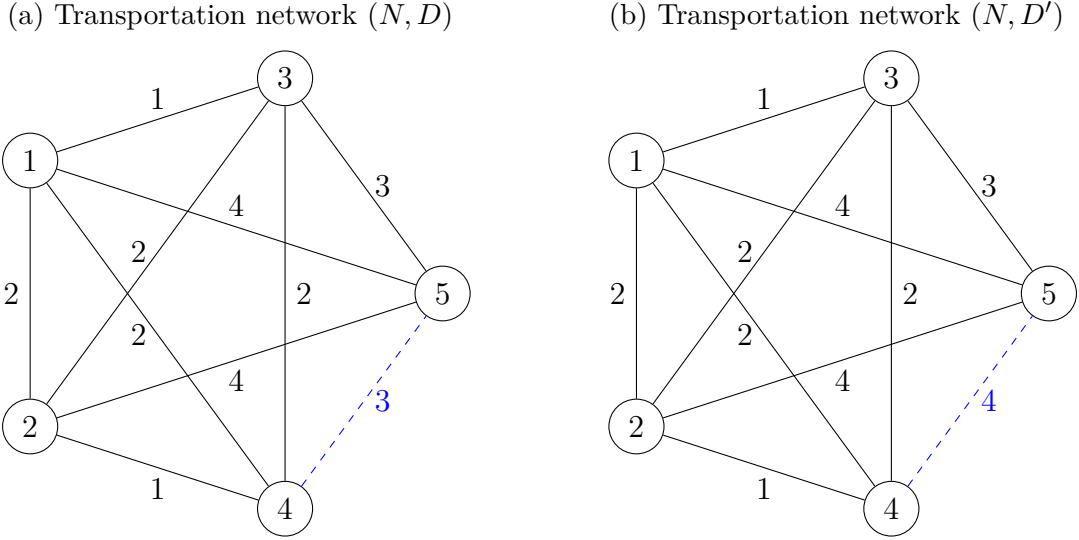

\begin{example} \label{Examp2}
Consider the transportation networks $(N,D), (N,D') \in \mathcal{N}^5$ in Figure~\ref{Fig2}.
$D'$ is the same as $D$ except for $d_{45}' = 4 > 3 = d_{45}$. 
Therefore $d_1^\Sigma(N,D) = d_2^\Sigma(N,D)$ and $d_1^\Sigma(N,D) = d_2^\Sigma(N,D)$ as well as $d_1^\Pi(N,D) = d_2^\Pi(N,D)$ and $d_1^\Pi(N,D) = d_2^\Pi(N,D)$
\end{example}

We think that -- together with their distance -- the accessibility of other nodes also count: in the case of two nodes with the same distance sum, it is better to be close to more accessible nodes (and therefore, to be far to less accessible ones) than vice versa.

\begin{notation} \label{Not4}
Let $(N,D) \in \mathcal{N}^n$ be a transportation network. Matrix $A = \left( a_{ij} \right) \in \mathbb{R}^{n \times n}$ is given by $a_{ij} = d_{ij} / d_i^\Sigma(N,D)$ for all $i \neq j$ and $a_{ii} = -\sum_{j \neq i} d_{ij} / d_j^\Sigma(N,D)$ for all $i \in N$.
\end{notation}

The sum of a column of $A$ is zero. The sum of a row of $A$ without the diagonal element is one.

\begin{definition} \label{Def14}
\emph{Generalized distance sum}:
$\mathbf{x}(\alpha): \mathcal{N}^n \to \mathbb{R}^n$ such that $(I+ \alpha A) \mathbf{x}(\alpha) = \mathbf{d^\Sigma}$, where $\alpha > 0$ is a parameter.
\end{definition}

\begin{remark} \label{Rem4}
Generalized distance sum can be expressed in the following form for all $i \in N$:
\[
\left( 1- \alpha \sum_{j \neq i} \frac{d_{ij}}{d_j^\Sigma} \right) x_i(\alpha) = \left( d_i^\Sigma - \alpha \sum_{j \neq i} \frac{d_{ij}}{d_i^\Sigma} x_j(\alpha) \right).
\]
\end{remark}

\begin{lemma} \label{Lemma5}
Generalized distance sum does not meet independence of irrelevant distances.
\end{lemma}

\input{Figure3}

\begin{proof}
\begin{example} \label{Examp3}
Consider the transportation networks $(N,D), (N,D') \in \mathcal{N}^5$ in Figure~\ref{Fig2}. 
Generalized distance sums with various values of $\alpha$ are given in Figure~\ref{Fig3}.

Nodes $1$ and $2$ as well as $3$ and $4$ are symmetric in $(N,D)$, while $2$ and $4$ are symmetric in $(N,D')$, they have the same generalized distance sum for any $\alpha$. Generalized distance sums monotonically decrease with the exception of node $5$ (see Figures~\ref{Fig3a} and \ref{Fig3b}, where the numbers on the curves indicate the corresponding node).
According to Figure~\ref{Fig3d}, it can be achieved by increasing $\alpha$ that $x_1(\alpha) < x_3(\alpha)$ (node $3$ is less accessible than node $1$) both in $(N,D)$ and $(N,D')$.

The aim was to demonstrate that node $1$ is more accessible than node $2$ in $(N,D')$ since they have the same distance sum but the former is closer to node $3$ than the latter, and node $3$ seems to be more accessible than node $4$ as they were symmetric in $(N,D)$ and $d_{45}$ has increased. This relation holds if $\alpha$ is not too large (see Figures~\ref{Fig3c} and \ref{Fig3d}).
\end{example}
\end{proof}

Some basic attributes of generalized distance sum are listed below.

\begin{proposition} \label{Prop2}
Generalized distance sum satisfies the following properties for any fixed parameter $0 < \alpha < -1 / \min \left\{ \lambda: \lambda A = A \mathbf{y} \right\}$ and transportation network $(N,D) \in \mathcal{N}$:
\begin{enumerate}
\item
\emph{Existence and uniqueness}: a unique vector of $\mathbf{x}(\alpha)$ exists;

\item
\emph{Anonymity} ($ANO$);

\item
\emph{Homogeneity} ($HOM$): the relation
\[
x_i(\alpha)(N,D) \leq x_i(\alpha)(N,D) \Rightarrow x_i(\alpha)(N,\beta D) \leq x_i(\alpha)(N,\beta D)
\]
holds for all $i,j \in N$ and $\beta > 0$;

\item
\emph{Distance sum conservation}: $\sum_{i \in N} x_i(\alpha) = \sum_{i \in N} d_i^\Sigma$;

\item
\emph{Agreement}: $\lim_{\alpha \to 0} \mathbf{x}(\alpha) = \mathbf{d^\Sigma}$;

\item
\emph{Flatness preservation} ($FP$): if $d_i^\Sigma = d_j^\Sigma$ for all $i,j \in N$, then $x_i(\alpha) = x_j(\alpha)$ for all $i,j \in N$.
\end{enumerate}

\end{proposition}

\begin{proof}
We prove the statements above in the corresponding order.
\begin{enumerate}
\item
Matrix $(I + \alpha A)$ is positive definite if $0 < \alpha < -1 / \min \left\{ \lambda: \lambda A = A \mathbf{y} \right\}$.

\item
Generalized distance sum is invariant under isomorphism, it depends just on the structure of the transportation network and not on the labelling of the nodes.

\item
The identities $\mathbf{d^\Sigma}(N,\beta D) = \beta \mathbf{d^\Sigma}(N,D)$ and $A(N,\beta D) = A(N,D)$ imply \linebreak $\mathbf{x}(\alpha)(N,\beta D) = \beta \mathbf{x}(\alpha)(N,D)$.

\item
$\sum_{i \in N} x_i(\alpha) = \sum_{i \in N} \left[ x_i(\alpha) + \sum_{j \in N} a_{ij} x_j(\alpha) \right] = \sum_{i \in N} d_i^\Sigma$ since $\sum_{i \in N} a_{ij} = 0$ for all $j \in N$.

\item
$\lim_{\alpha \to 0} \alpha A \mathbf{x}(\alpha) = \mathbf{0}$.

\item
$x_i(\alpha) = d_i^\Sigma$ satisfies $(I+ \alpha A) \mathbf{x}(\alpha) = \mathbf{d^\Sigma}$ in the case of $d_i^\Sigma = d_j^\Sigma$ for all $i,j \in N$. \\
\end{enumerate}
\end{proof}

Regarding $ANO$, it remains an interesting question whether the opposite of symmetry is valid, i.e., two nodes have the same generalized distance sum for any $\alpha > 0$ only if they are symmetric.

Homogeneity means that the accessibility ranking is invariant under the multiplication of all distances by a positive scalar (i.e. the choice of measurement scale).
Note that triangle inequality is preserved when distances are multiplied by a positive scalar.

The name of this accessibility index comes from the property distance sum conservation: it can be interpreted as a way to redistribute the sum of distances among the nodes. This axiom contributes to the comparability of calculated accessibilities across different networks.
According to agreement, its limit is distance sum.

$FP$ requires that generalized distance sum results in a tied accessibility between any nodes if distance sum also gives this result. The other direction (whether the flatness of generalized distance for a given $\alpha$ implies the flatness of distance sum) requires further research.

\begin{remark} \label{Rem5}
Distance sum also satisfies the properties listed in Proposition~\ref{Prop2} (with the obvious exception of agreement).
Distance product does not meet distance sum conservation and flatness preservation.
\end{remark}


Triangle inequality is essential for an accessibility index satisfying $FP$.

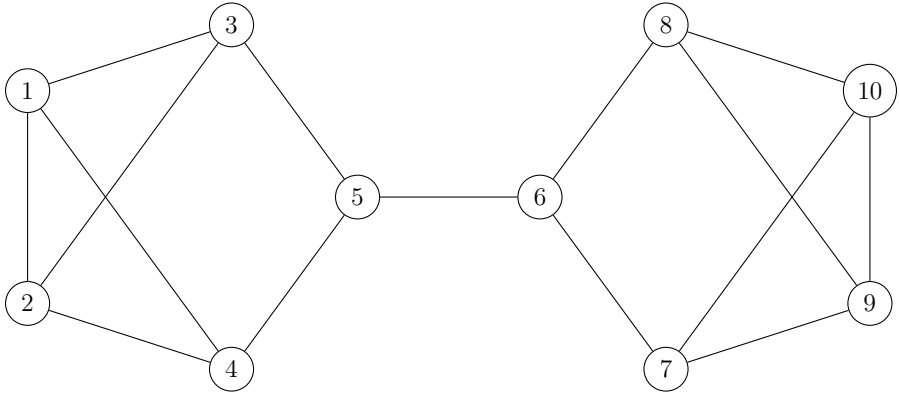
\begin{figure}[!ht]
\centering
\caption{Transportation network of Example~\ref{Examp4}}
\label{Fig4}

\begin{tikzpicture}[scale=0.8, auto=center, transform shape, >=triangle 45]
\tikzstyle{every node}=[draw,shape=circle];
  \node (n1) at (144:3) {$1$};
  \node (n2) at (216:3) {$2$};
  \node (n3) at (72:3)  {$3$};
  \node (n4) at (288:3) {$4$};
  \node (n5) at (0:3)   {$5$};

  \foreach \from/ \to in {n1/n2,n1/n3,n1/n4,n2/n3,n2/n4,n3/n5,n4/n5}
    \draw (\from) -- (\to);
 
\begin{scope}[xshift=9cm]
  \node (m1) at (36:3)  {$10$};
  \node (m2) at (324:3) {$9$};
  \node (m3) at (108:3) {$8$};
  \node (m4) at (252:3) {$7$};
  \node (m5) at (180:3) {$6$};

  \foreach \from/ \to in {m1/m2,m1/m3,m1/m4,m2/m3,m2/m4,m3/m5,m4/m5}
    \draw (\from) -- (\to);
  
  \draw (n5) -- (m5);
\end{scope}
\end{tikzpicture}
\end{figure}

\begin{example} \label{Examp4}
Consider the transportation network $(N,D) \in \mathcal{N}^{10}$ in Figure~\ref{Fig4} such that
\begin{itemize}
\item
$d_{ij} = 1$ if and only if nodes $i$ and $j$ are connected;
\item
$d_{ij} = 10$ if and only if nodes $i$ and $j$ are \emph{not} connected.
\end{itemize}
Then $d_i^\Sigma = 3 \times 1 + 6 \times 10 = 63$ for all $i \in N$, but nodes $5$ and $6$ seems to be more accessible than the others. It is the case until $d_{ij} > 2$ in the case of not connected nodes $i$ and $j$ since then there exists more shortcuts for nodes $5$ and $6$ than for the others.
\end{example}

It can be seen (e.g. from Example~\ref{Fig3}) that generalized distance sum is better able to distinguish the accessibility of nodes than distance sum. In other words, there are far less nodes of equal generalized distance sum than there are for standard distance sum. It may be important in some applications \citep{LindelaufHamersHusslage2013}.

\subsection{Connection to domination preservation} \label{Sec32}

Generalized distance sum was constructed in order to eliminate independence of distance distribution. Proposition~\ref{Prop2} shows that anonymity is preserved in the process.
Now we turn to the third axiom in the characterization of distance sum by Theorem~\ref{Theo1}, and investigate whether the suggested accessibility index meets domination preservation.

\begin{notation} \label{Not5}
$\min_k \{ S \}$ and $\max_k \{ S \}$ are the $k$th smallest and largest element of a set $S$, respectively.
\end{notation}

\begin{theorem} \label{Theo2}
Generalized distance sum satisfies $DP$ if
\[
\alpha < \min \left\{ \frac{\min_2 \left\{ d_i^\Sigma: i \in N \right\} }{\max \left\{ x_i(\alpha): i \in N \right\} };\, \frac{\min \left\{ d_i^\Sigma: i \in N \right\} }{\max \left\{ d_i^\Sigma: i \in N \right\} } \right\}.
\]
\end{theorem}

\begin{proof}
Consider a transportation network $(N,D) \in \mathcal{N}$ and two distinct nodes $i,j \in N$ such that $d_{ik} \leq d_{jk}$ for all $k \in N \setminus \{ i,j \}$ with a strict inequality ($<$) for at least one $k$. Then
\[
x_i(\alpha) = \left( 1- \alpha \sum_{k \neq i} \frac{d_{ik}}{d_k^\Sigma} \right)^{-1} \left[ \sum_{k \neq i} d_{ik} \left( 1 - \alpha \frac{x_k(\alpha)}{d_i^\Sigma} \right) \right] \quad \text{and}
\]
\[
x_j(\alpha) = \left( 1- \alpha \sum_{k \neq j} \frac{d_{jk}}{d_k^\Sigma} \right)^{-1} \left[ \sum_{k \neq j} d_{jk} \left( 1 - \alpha \frac{x_k(\alpha)}{d_j^\Sigma} \right) \right].
\]
The following two conditions are sufficient for $x_i(\alpha) < x_j(\alpha)$:
\begin{itemize}
\item
$\alpha x_k(\alpha) / d_j^\Sigma < 1$ for any $k \in N \setminus \{ i,j \}$, that is, $\alpha < d_j^\Sigma / x_k(\alpha)$, which provides that the 'contribution' of $d_{jk}$ to $x_j(\alpha)$ is positive;

\item
$\alpha \sum_{k \neq i} d_{ik} / d_k^\Sigma \leq \alpha \sum_{k \neq j} d_{jk} / d_k^\Sigma \leq \alpha d_j^\Sigma / \min \left\{ d_\ell^\Sigma: \ell \in N \right\} < 1$, namely, $\alpha < \min \left\{ d_\ell^\Sigma: \ell \in N \right\} / d_j^\Sigma$. 
\end{itemize}
Note also that $d_j^\Sigma \geq \min_2 \left\{ d_\ell^\Sigma: \ell \in N \right\}$.
\end{proof}

\begin{remark} \label{Rem6}
According to the proof of Theorem~\ref{Theo2}, the condition
\[
\alpha < \min \left\{ d_i^\Sigma: i \in N \right\} / \max \left\{ x_i(\alpha): i \in N \right\}
\]
provides that the 'contribution' of any distance to any node's generalized distance sum is positive, therefore it is positive, too (look at Figure~\ref{Fig3d}).
\end{remark}

\begin{remark} \label{Rem7}
Due to the continuity of $\mathbf{x}(\alpha)$, $\min_2 \left\{ d_i^\Sigma: i \in N \right\} > \min \left\{ d_i^\Sigma: i \in N \right\}$ implies
\[
\frac{\min_2 \left\{ d_i^\Sigma: i \in N \right\} }{\max \left\{ x_i(\alpha): i \in N \right\} } > \frac{\min \left\{ d_i^\Sigma: i \in N \right\} }{\max \left\{ d_i^\Sigma: i \in N \right\} },
\]
for small $\alpha$, so the second expression is effective for small $\alpha$. Nevertheless, usually the first expression limits the value of the parameter $\alpha$.
\end{remark}

In a certain sense, the result of Theorem~\ref{Theo2} is not surprising because distance sum meets $DP$ and $\mathbf{x}(\alpha)$ is 'close' to it when $\alpha$ is small.\footnote{~There always exists an appropriately small $\alpha$ satisfying the condition of Theorem~\ref{Theo2}.}
The main contribution is the calculation of a sufficient condition. It is not elegant since it depends on $\mathbf{x}(\alpha)$, which excludes to check the requirement before the calculation,. However, ex post check is not difficult.

\begin{definition} \label{Def15}
\emph{Reasonable upper bound}:
Parameter $\alpha$ is \emph{reasonable} if it satisfies the condition of Theorem~\ref{Theo2} for a transportation network $(N,D) \in \mathcal{N}$. The largest $\alpha$ with this property is the \emph{reasonable upper bound} of the parameter.
\end{definition}

\begin{notation} \label{Not6}
The reasonable upper bound of generalized distance sum's parameter is $\hat{\alpha}$.
\end{notation}

\begin{remark} \label{Rem8}
Generalized degree satisfies $DP$ for any reasonable $\alpha > 0$.
\end{remark}

It is analogous to the (dynamic) monotonicity of a preference aggregation method called generalized row sum  \citep[Property~13]{Chebotarev1994} as well as to the adding rank monotonicity of a centrality measure generalized degree \citep{Csato2015c}.

Violation of domination preservation may be a problem in practice.

\begin{example} \label{Examp5}
Consider the transportation network $(N,D') \in \mathcal{N}^5$ in Figure~\ref{Fig2}, where node $3$ dominates node $1$ since $d_{3k} \leq d_{1k}$ for $k = 2,4$ and $d_{35} < d_{15}$.
However, $x_3(0.4) > x_1(0.4)$ according to Figure~\ref{Fig3d}.

Theorem~\ref{Theo2} gives the reasonable upper bound as $\hat{\alpha} \approx 0.2686$ from
\[
\alpha < \frac{\min_2 \left\{ d_i^\Sigma: i \in N \right\} }{\max \left\{ x_i(\alpha): i \in N \right\} }.
\]
Here $\min \left\{ d_i^\Sigma: i \in N \right\} / \max \left\{ d_i^\Sigma: i \in N \right\} = 8/15$.
However, Theorem~\ref{Theo2} does not give a necessary condition for $DP$ since, for example, $x_3(0.36) < x_1(0.36)$ in Figure~\ref{Fig3d}.
\end{example}

Instead of domination preservation, which is a kind of 'static' axiom, one can focus on the 'dynamic' monotonicity properties of accessibility indices, similarly to independence of distance distribution of independence or irrelevant distances. This direction is followed by \citet{Sabidussi1966} or \citet{Chienetal2004} in the case of centrality measures.
Then a plausible condition is the following.

\begin{definition} \label{Def16}
\emph{Positive responsiveness to distances} ($PRD$):
Let $(N,D) \in \mathcal{N}^n$ be a transportation network and $i,j,k \in N$ be three distinct nodes.
Let $f: \mathcal{N}^n \to \mathbb{R}^n$ be an accessibility index such that $f_i(N,D) \leq f_j(N,D)$ and $(N,D') \in \mathcal{N}^n$ be a transportation network identical to $(N,D)$ except for $d'_{jk} > d_{jk}$. \\
$f$ is called \emph{positively responsive to distances} if $f_i(N,D) < f_j(N,D)$.
\end{definition}

Property $PRD$ demand that the position of a node in the accessibility ranking does not improve after an increase in its distances to other nodes. 
It has strong links to dominance preservation.

\begin{lemma} \label{Lemma6}
$ANO$ and $PRD$ implies $DP$.
\end{lemma}

\begin{proof}
Consider a transportation network $(N,D) \in \mathcal{N}$ and two distinct nodes $i,j \in N$ such that $d_{ik} \leq d_{jk}$ for all $k \in N \setminus \{ i,j \}$ with a strict inequality ($<$) for at least one $k$. Define transportation network $(N,D')$, which is the same as $(N,D)$ except for $d_{ik} = d_{jk}$ for all $k \in N \setminus \{ i,j \}$. Anonymity provides that $f_i(N,D) = f_j(N,D)$, so positive responsiveness to distances implies $f_i(N,D) < f_j(N,D)$.
\end{proof}

\begin{lemma} \label{Lemma7}
Distance sum, distance sum without ties and distance product meet $PRD$.
\end{lemma}

\begin{remark} \label{Rem9}
Analogously to Theorem~\ref{Theo1}, distance sum can be characterized by $ANO$, $IDD$ and $PRD$. Their independence is shown by the same three accessibility indices distance sum without ties, inverse distance sum, and distance product.
\end{remark}

\begin{proposition} \label{Prop3}
There exists an anonymous accessibility index preserving dominance, which is not positively responsive to distances.
\end{proposition}

\begin{proof}
It is enough to define a function which gives an accessibility ranking of the nodes.

\begin{definition} \label{Def17}
\emph{Lexicographic eccentricity}: node $i$ is more accessible than node $j$ if and only if $\max \left\{ d_{ik}: k \in  N \right\} < \max \left\{ d_{jk}: k \in  N \right\}$ or $\max_m \left\{ d_{ik}: k \in  N \right\} < \max_m \left\{ d_{jk}: k \in  N \right\}$ and $\max_\ell \left\{ d_{ik}: k \in  N \right\} = \max_\ell \left\{ d_{jk}: k \in  N \right\}$ for all $\ell = 1,2, \dots, m-1$.
\end{definition}

It is the same as eccentricity \citep{HageHarary1995}, however, ties are broken by a lexicographic application of its criteria.
Lexicographic eccentricity satisfies $ANO$ and $DP$, but it is not positively responsive to distances since some distances in certain intervals has no effect on the accessibility ranking. 
\end{proof}

\begin{remark} \label{Rem10}
Remark~\ref{Rem9} and Proposition~\ref{Prop3} show a somewhat surprising feature of axiomatizations. Distance sum can be characterized by the three independent axioms of $ANO$, $IDD$ and $PRD$ as well as by $ANO$, $IDD$ and $DP$. However, $DP$ is implied by $ANO$ and $PRD$. Therefore the second result seems to be stronger (as $DP$ is strictly weaker than the combination of $ANO$ and $PRD$) but it is by the first axiomatization.
\end{remark}

It remains an open question to give a sufficient condition for generalized distance sum to satisfy $PRD$.

\subsection{Two detailed examples} \label{Sec33}

In the previous discussion, generalized distance sum was only a 'tie-breaking rule' of distance sum in the reasonable interval of the parameter $\alpha$. Therefore, it remains to be seen whether the choice of a reasonable $\alpha$ may lead to significant changes in the accessibility ranking. The following examples reveal some interesting features of the measure, too.

\input{Figure5}

\begin{example} \label{Examp6}
The Marshall Islands in eastern Micronesia are divided into two atoll chains, one of them is Ralik. Figure~\ref{Fig5a} shows the graph of the voyaging network of this island chain, constructed by \citet{HageHarary1995}. The distances of islands are measured by the shortest path between them. For instance, the distance of Bikini and Jaluit is $5$ through Rongelap, Kwajalein, Namu and Ailinglaplap.

Note that Jaluit and Namorik are structurally equivalent in the network. Furthermore, Wotho dominates Rongelap and Ujae (it has extra links to Ujau and Lae, and to Bikini and Rongelap, respectively); Wotho and Rongelap dominate Bikini; Lae dominates Ujae; Ailinglaplap, Namorik and Jaluit dominate Ebon.

Generalized distance sums are presented in Figure~\ref{Fig5b}. The reasonable upper bound is $\hat{\alpha} \approx 0.2607$, indicated by a (black) vertical line. However, the accessibility ranking does not violate $DP$ for any $\alpha$ in Figure~\ref{Fig5b}.

Kwajalein (with a distance sum of $20$) and Namu ($21$) are the first and second nodes in the accessibility ranking for any $\alpha$. The difference between their generalized distance sums monotonically increases.
Ebon ($34$) is 'obviously' the least accessible node for any $\alpha$.
These nodes do not appear in Figure~\ref{Fig5b}. 

Compared to the distance sum, the following changes can be observed in the reasonable interval of $\alpha$:
\begin{itemize}
\item
The tie between Ailinglaplap and Wotho ($25$) is broken for Wotho. It makes sense since the nodes around Wotho have more links among them.

\item
The tie between Rongelap and Ujae ($27$) is broken for Ujae. It is justified by Ujae's direct connection to Lae instead of Bikini as the former is more accessible than the latter.

\item
Lae ($26$) overtakes Ailinglaplap ($25$). The cause is that the network has essentially two components: the link between Ailinglaplap and Namu is a cut-edge, and the above part around Kwajalein (where Lae is located) is bigger and has more internal links.
\end{itemize}
Some other changes occur outside the reasonable interval of $\alpha$: Rongelap and Ujau overtake Ailinglaplap as well as Bikini ($34$) overtakes Jaluit and Namorik ($32$). The reasoning applied for the case of Ailinglaplap and Lae is relevant here, too, and it is difficult to argue against these modifications of the ranking.
\end{example}

Example~\ref{Examp6} verifies that the analysis of accessibility rankings should not automatically limit to the reasonable interval of parameter $\alpha$, sometimes it is worth to consider values outside it.
Remember that the condition of Theorem~\ref{Theo2} is only sufficient, but not necessary for dominance preservation.

\input{Figure6}

\begin{example} \label{Examp7}
Consider the network in Figure~\ref{Fig6a}. The transportation network $(N,D) \in \mathcal{N}^{12}$ is such that the distances of nodes are measured by the shortest path between them. For example, $d_{1,12} = 3$ due to the path $(1,2)-(2,11)-(11,12)$.

Here node $1$ dominates nodes $3$, $8$ and $9$, node $2$ dominates node $11$, node $3$ dominates node $9$, node $5$ dominates node $10$ and node $11$ dominates node $12$.

Generalized distance sums are presented in Figure~\ref{Fig6b}. The reasonable upper bound is $\hat{\alpha} \approx 0.279$, indicated by a (black) vertical line. However, the accessibility ranking does not violate $DP$ for any $\alpha$ in Figure~\ref{Fig6b}.

Node $1$ ($d_1^{\Sigma} = 17$) is 'obviously' the most accessible node for any $\alpha$.
Nodes $11$ ($d_{11}^{\Sigma} = 29$) and $12$ ($d_{12}^{\Sigma} = 37$) are the last in the accessibility ranking for any reasonable value of $\alpha$. The difference between their generalized distance sums monotonically increases.
These nodes do not appear in Figure~\ref{Fig6b}.

Note that $x_{10}(\alpha)$ is not monotonic, it has a maximum around $\alpha = 0.145$.

Compared to the distance sum, the following changes can be observed in the reasonable interval of $\alpha$:
\begin{itemize}
\item
The tie between nodes $3$, $4$, $5$ and $7$ ($d^{\Sigma} = 22$) is broken in this order.
Node $3$ is more accessible than node $5$ since node $9$ is more accessible than node $10$. Node $3$ is more accessible than node $7$ since node $4$ is more accessible than node $8$ (and it has an extra link to node $9$).
Node $4$ is more accessible than node $5$ since node $3$ is more accessible than node $6$. Node $4$ is more accessible than node $7$ as its neighbours are more accessible.
Node $5$ is more accessible than node $7$ since node $4$ is more accessible than node $8$ (and it has an extra link to node $10$).
The tie breaks are highlighted in Figure~\ref{Fig6c}.

\item
The relation of nodes $3$ and $4$ ($d^{\Sigma} = 22$) is uncertain: both are connected to node $1$, while node $5$ is more accessible than node $6$ (favouring node $4$) and node $9$ is more accessible than node $10$ (favouring node $3$). Node $3$ benefits from a small $\alpha$ as shown in Figure~\ref{Fig6d}.
However, one can see the negligible difference of generalized distance sums, red and purple lines are indistinguishable on Figures~\ref{Fig6b} and \ref{Fig6c}.

\item
Node $10$ ($d_{10}^{\Sigma} = 27$) overtakes node $8$ ($d_{8}^{\Sigma} = 26$). It is more connected to the above 'component' of the network. Furthermore, node $8$ is closer to node $2$, which gradually losses positions in the accessibility ranking.

\item
Nodes $3$, $4$, $5$ and $7$ ($d^{\Sigma} = 22$) overtake node $2$ ($d_2^{\Sigma} = 21$). It is explained by the latter's vulnerable connections: if the link to node $1$ is eliminated, then node $2$ suffers severe deterioration in its accessibility.
\end{itemize}
For larger $\alpha$, node $6$ (and even nodes $9$ and $10$) become(s) more accessible than node $2$, too. 
\end{example}

Example~\ref{Examp7} illustrates that generalized distance sum may grab the vulnerability of some nodes' accessibility to link disruptions.
It can also reveal -- analogously to the non-principal eigenvectors of the adjacency matrix -- geographical subsystems of the network.

It is essential that any changes of the accessibility ranking in Examples~\ref{Examp6} and \ref{Examp7} may be attributable to the structure of the network graph.

\section{Concluding remarks} \label{Sec4}

Measuring accessibility is an important issue in network analysis. This paper has aimed to explore the methodological background of some accessibility indices for networks represented by a complete graph, where each link has a value such that a smaller number is preferred (e.g. distance, cost, or travel time) and triangle inequality holds.

The obvious solution of distance sum is characterized by three independent axioms, anonymity, independence of distance distributions and dominance preservation. Generalized distance sum, a parametric family of linear accessibility indices, is suggested in order to eliminate independence of distance distribution, a condition one would rather not have in certain applications. It considers the accessibility of vertices besides their distances (and gives an infinite depth to this argument) and depends on a parameter in order to control its deviation from distance sum.
Generalized distance sum seems to be promising with respect to its properties, has a good level of differentiation, and gives acceptable results on several transportation networks by making accessibility responsive to link disruptions.

Generalized distance sum -- similarly to connection array -- may get criticism for depending on a seemingly arbitrary value. While it certainly causes some difficulties in applications, parametrization is necessary to ensure the flexibility of the method. Alternative rankings should be regarded as revealing the uncertainty of accessibility. A research pursuing normative goals could not assume the complex task of making decisions, 

Although the suggested accessibility index offers a new way in the analysis of transportation networks, some issues have remained unanswered.
The choice of parameter requires further investigation and the characterization of generalized distance sum is worth a try. The behaviour of this index should be tested on other networks.
Generalized distance sum may provide a basis for measuring the total accessibility of the network by a single number. 
Possible domain extensions include networks where the importance of nodes is different, or the distance matrix is not symmetric.
They will be considered for future research.


\end{document}